\documentclass{amc}
\usepackage{hyperref}
 \usepackage{amsmath}
 \usepackage{epsfig} 
 \usepackage{graphics} 

\def\currentvolume{1}
 \def\currentissue{3}
  \def\currentyear{2007}
   
    \def\ppages{287--307}

\newtheorem{theorem}{Theorem}

\newtheorem{lemma}{Lemma}

\newtheorem{conjecture}{Conjecture}

\theoremstyle{definition}
\newtheorem{definition}{Definition}
\newtheorem{remark}{Remark}

\title[Eigenvalue bounds on the pseudocodeword weight of expander codes]
      {Eigenvalue bounds on the pseudocodeword weight of expander codes}

\author[Christine A. Kelley and Deepak Sridhara]{}

\subjclass{Primary: 58F15, 58F17; Secondary: 53C35}
 \keywords{Expander graphs, LDPC codes, pseudocodewords,
  pseudocodeword weight, iterative decoding}

\thanks{The first author is with the Department of Mathematics at The Ohio State University. She was
previously with the Fields Institute, Toronto, Canada. The second
author is with Seagate Technology, Pittsburgh, USA. He was
previously with the Institut f\"ur Mathematik, Universit\"at
Z\"urich, Switzerland. This work was supported in part by the Swiss
National Science Foundation under Grant No. 113251.}

\begin{document}
\maketitle

\centerline{\scshape Christine A. Kelley}

\medskip

{\footnotesize
  \centerline{Department of Mathematics }
   \centerline{The Ohio State University}
    \centerline{Columbus, OH 43210, USA}
}

\medskip

\centerline{\scshape Deepak Sridhara}

\medskip

{\footnotesize
  \centerline{Seagate Technology}
  \centerline{1251 Waterfront Place}
    \centerline{ Pittsburgh, PA 15222, USA}

} %

\bigskip

 \centerline{(Communicated by Marcus Greferath)}
 \medskip

\begin{abstract}
Four different ways of obtaining low-density parity-check codes from expander
graphs are considered. For each case, lower bounds on the minimum stopping set size and the
minimum pseudocodeword weight of expander (LDPC) codes are derived. These
bounds are compared with the known eigenvalue-based lower bounds on the
minimum distance of expander codes. Furthermore, Tanner's parity-oriented
eigenvalue lower bound on the minimum distance is generalized to yield a
new lower bound on the minimum pseudocodeword weight. These bounds are useful in predicting
the performance of LDPC codes under graph-based iterative decoding and linear programming decoding.
\end{abstract}

\section{Introduction}
\label{intro}
Expander graphs are of fundamental interest in mathematics and engineering and have several applications in
computer science, complexity theory, derandomization, designing communication networks, and coding theory \cite{li03r,al86}.
A family of highly expanding graphs known as {\rm  Ramanujan} graphs \cite{lu88a,ma82} was constructed
with excellent graph properties that surpassed the parameters predicted for random graphs.
The description of these graphs and their analysis rely on deep results from mathematics using tools from graph theory,
number theory, and representation theory of groups \cite{lu94b}. Other authors have investigated non-algebraic
approaches to designing expander graphs and one such construction takes an appropriately defined product of small component expander
graphs to construct a larger expander graph \cite{li03r,al01p,rvw02}. Moreover, expander graphs have a special appeal
from a geometric viewpoint. Isoperimetric problems in geometry have also been described by analogous problems in
graphs, and a close connection exists between the Cheeger constant, defined for Riemannian surfaces, and the expansion constant
in graphs. Expander graphs can be viewed as discrete analogues of
Riemannian manifolds.

In this paper, we focus on one prominent application of expander graphs -- namely, the design of low-density parity-check (LDPC) codes.
Low-density parity-check codes are a class of codes that can be represented on
sparse graphs and have been shown to achieve record breaking performances
with graph-based message-passing decoders. Graphs with good expansion properties are particularly suited
for the decoder in dispersing messages to all nodes in the graph as quickly as possible. Expander codes are families
of graph-based codes where the underlying graphs are expanders. That is,
every element
of the family is an expander and gives rise to an expander code. The codes are
obtained by imposing code-constraints on the vertices (and possibly, edges)
of the underlying expander graphs \cite{si96,la00p,ja03}. It has been
observed that
graphs with good expansion lead to LDPC codes with minimum distance\footnote{The minimum distance of a code is a fundamental parameter that determines its error-correction capability.} growing linearly with the block length. In fact, one method of
designing asymptotically good linear block codes is from expander graphs \cite{si96}.

The popularity of LDPC codes is that they can be decoded with linear time complexity using graph-based
message-passing decoders, thereby allowing for the use of large block length codes in several practical
applications. In contrast, maximum-likelihood (ML) decoding a generic error-correcting code is known to be NP hard.
A parameter that dominates the performance of a graph-based message passing decoder is the minimum
pseudocodeword weight, in contrast to the minimum distance for an optimal (or, ML) decoder.
The minimum pseudocodeword weight of the graph has been found to be a reasonable
predictor of the performance of a finite-length LDPC code under graph-based message-passing decoding and also linear
programming decoding \cite{ko03p,ke05u,feldman,fe03t,fe05}.
In this paper, we consider four different ways of obtaining LDPC codes from expander graphs. For each case, we first
present the known lower bounds on the minimum distance of expander codes based
on the expansion properties of the underlying expander graph. We then extend the results to lower bound
the minimum stopping set size, which is essentially the minimum pseudocodeword weight on the binary erasure
channel (BEC), and finally, we lower bound the minimum pseudocodeword weight on the binary symmetric channel (BSC).
We also examine a new parity-oriented lower bound on the minimum pseudocodeword weight over the additive white
Gaussian noise (AWGN) channel, thereby generalizing the result of Tanner
\cite{ta01} for the minimum distance.

\section{Preliminaries}
\label{sec:1}
We introduce some preliminary definitions and notation that we will use in this paper.

\begin{definition}{\rm
A graph $G=(X,Y;E)$ is {\em $(c,d)$-regular bipartite} if the set of
vertices in $G$ can be partitioned
into two disjoint sets $X$ and $Y$ such that all vertices in
$X$ have degree $c$ and all vertices in $Y$ have degree $d$ and each edge  $e\in E$
of $G$ is incident with one vertex in $X$ and one vertex in $Y$, i.e.,
$e=(x,y), x\in X, y\in Y$.}
\end{definition}

We will refer to the vertices of degree $c$ as the {\em left} vertices,
and to vertices of degree $d$ as the {\em right} vertices.


The adjacency matrix of a $d$-regular connected graph has $d$ as its
largest eigenvalue. Informally, a graph is a good expander if the gap
between the first and the second largest eigenvalues of the adjacency matrix is as big as
possible. More precise definitions will be given later in the
paper as needed. Note that for a $(c,d)$-regular bipartite graph,
the largest eigenvalue is $\sqrt{cd}$.


\begin{definition}
{\rm
A {\em simple} LDPC code is defined by a bipartite graph $G$ (also called, a Tanner graph) whose
left vertices are called {\em variable}
(or, {\em codebit}) nodes and whose right vertices are called {\em check}
(or, {\em constraint}) nodes and the set
of codewords are all binary assignments to the variable nodes such that
at each check node, the modulo-two sum
of the variable node assignments connected to the check node is zero,
i.e., the parity-check constraint involving the
neighboring variable nodes is satisfied.
}
\label{LDPC_defn}
\end{definition}
Note that equivalently, the LDPC code can be described by a (binary)
incidence matrix (or, parity-check matrix) wherein the rows of the matrix
correspond to the constraint nodes of $G$ and the columns correspond to
variable nodes and there is a one in the matrix at a row-column entry
whenever there is an edge between the corresponding constraint node and
variable node in $G$.

The above definition can be generalized by introducing more complex
constraints instead of simple parity-check constraints at each constraint
node, and the resulting LDPC code will be called a {\em generalized} LDPC
code.

To analyze the performance of graph-based message passing decoding, certain combinatorial objects of the LDPC constraint
graph have been identified that control the performance of the decoder. When transmitting over a binary erasure channel (BEC), it has been shown that stopping sets in the Tanner graph control the performance of the message-passing decoder.

\begin{definition}{\rm \cite{di02}
For a simple LDPC code, a {\em stopping set} is a subset set $S$ of the variable nodes such that every constraint node that is a neighbor of some node $s\in S$
is connected to $S$ at least twice.}
\end{definition}
\vspace{0.1in}

\begin{figure}
\centering{\resizebox{4in}{1.15in}{\includegraphics{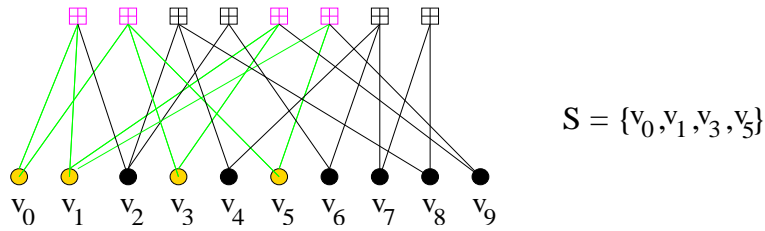}}}
\caption{A stopping set $S = \{v_0,v_1,v_3,v_5\}$ in $G$.}
\label{sset_fig1}
\end{figure}

The size of a stopping set $S$ is equal to the number of elements in $S$.
A stopping set is said to be {\em minimal} if there is no smaller sized
stopping set contained within it. The smallest minimal stopping set is
 called a {\em minimum} stopping set, and its size is denoted by
$s_{\min}$. Note that a minimum stopping set is not necessarily unique.
Figure~\ref{sset_fig1} shows a stopping set in the graph. Observe that
$\{v_4,v_7,v_8\}$ and $ \{v_3,v_5,v_9\}$ are two minimum stopping sets of
size $s_{\min} = 3$, whereas $\{v_0,v_1,v_3,v_5\}$ is a minimal stopping
set of size 4.

On the BEC, if all of the nodes of a stopping set are
erased, then the graph-based iterative
decoder will not be able to recover the erased symbols associated
with the nodes of the
stopping set \cite{di02}. Therefore, it is advantageous to design
LDPC codes with large minimum stopping set size $s_{\min}$.

For other channels, it has been recently observed that so called {\em pseudocodewords} dominate the performance of the
iterative decoder \cite{ko03p,ke05u}. (In fact, pseudocodewords are a generalization of stopping sets for other channels.)
We will now introduce the formal definition of {\em lift-realizable}
pseudocodewords of an LDPC constraint graph $G$ \cite{ke05u}.
However, we will also need to introduce the definition of a graph lift.
A degree $\ell$ cover (or, lift) $\hat{G}$ of $G$ is defined in the
following manner:

\begin{definition}{\rm A finite degree {\em $\ell$ cover} of $G=(V,W;E)$
is
a bipartite graph $\hat{G}$ where for each vertex $x_i\in V \cup W$,
there is a {\em
cloud}
$\hat{X}_i=\{\hat{x}_{i_1},\hat{x}_{i_2},\dots,\hat{x}_{i_{\ell}}\}$
of vertices in $\hat{G}$, with $deg(\hat{x}_{i_j})=deg(x_i)$ for all
$1\le j \le \ell$,
and for every $(x_i, x_j)\in E$, there are $\ell$ edges from $\hat{X}_i$
to $\hat{X}_j$ in $\hat{G}$ connected in a $1-1$ manner.
}\label{graphcover}
\end{definition}
Figure~\ref{pscw_fig} shows a base graph $G$ and a degree four cover of
$G$.

\begin{definition}{\rm Suppose that ${\bf \hat{\bf c}} =
({\hat{c}}_{1,1},{\hat{c}}_{1,2},\dots,{\hat{c}}_{1,\ell},{\hat{c}}_{2,1},\dots,{\hat{c}}_{2,\ell},\dots)$
is a codeword in the Tanner graph $\hat{G}$ representing a degree
$\ell$ cover of $G$. A {\em pseudocodeword ${\bf p}$ of $G$} is a
vector $(p_1,p_2,\dots,p_n)$ obtained by reducing a codeword ${\bf
\hat{\bf c}}$, of the code in the cover graph $\hat{G}$, in the
following way:
\begin{center}{ ${\bf
\hat{c}}=({\hat{c}}_{1,1},\dots,{\hat{c}}_{1,\ell},{\hat{c}}_{2,1},\dots,{\hat{c}}_{2,\ell},\dots)
\rightarrow  (\frac{{\hat{c}}_{1,1}+{\hat{c}}_{1,2}+
\dots+{\hat{c}}_{1,\ell}}{\ell},\frac{{\hat{c}}_{2,1}+{\hat{c}}_{2,2}+
\dots+{\hat{c}}_{2,\ell}}{\ell},\dots) = (p_1,p_2,\dots,p_n)$=${\bf p}$, }
\end{center}
\vspace{-0in}$\mbox{ where }p_i = \frac{{\hat{c}}_{i,1}+{\hat{c}}_{i,2}+
\dots+{\hat{c}}_{i,\ell}}{\ell}. $
}
\label{pscw_defn}
\end{definition}

The vector $\hat{c}$ on the left hand side of Figure ~\ref{pscw_fig} corresponds to a codeword in the degree four cover that is also
a codeword in the base graph $G$, whereas the vector on the right hand side corresponds to a codeword in the degree four cover that
does not correspond to a codeword in the base graph.

From the above definition, it is easy to show that for a simple LDPC constraint graph $G$, a pseudocodeword ${\bf p}=(p_1,p_2,\dots,p_n)$
is a vector that satisfies the following set of inequalities:
\begin{equation}
0\le p_i\le 1, \ \ \mbox{ for } i = 1,2,\dots,n.
\end{equation}
and, if variable nodes $i_1$, $i_2$, $\dots$, $i_d$ participate
in a check node of degree $d$, then the pseudocodeword components satisfy
\begin{equation}
p_{i_j}\le \sum_{k=1,2,..d, k\ne j} p_{i_k}, \mbox{ for } j=1,2,..,d.
\label{eqn_basic}
\end{equation}

Extending the above for generalized LDPC codes, it can similarly be shown that on a generalized LDPC constraint graph $G$,
a pseudocodeword ${\bf p}=(p_1,p_2,\dots,p_n)$ is a vector that satisfies the following set of inequalities:
\begin{equation}
0\le p_i\le 1, \ \ \mbox{ for } i = 1,2,\dots,n. \label{eqn_zero}
\end{equation}
and, if variable nodes $i_1$, $i_2$, $\dots$, $i_d$ participate
in a constraint node of degree $d$ and that constraint node represents a subcode $[d,rd,\epsilon d]$,  then the pseudocodeword components satisfy
\begin{equation}
(d\epsilon-1)p_{i_j}\le \sum_{k=1,2,..d, k\ne j} p_{i_k}, \mbox{ for } j=1,2,..,d.
\label{eqn_gen}
\end{equation}

\begin{remark}
Note that Equation \ref{eqn_gen} implies that the  pseudocodeword
components of the generalized LDPC constraint graph $G$ also satisfy
the following set of inequalities at a degree $d$ constraint node
representing a $[d,rd,\epsilon d]$ subcode
\begin{equation}
\sum_{\mbox {any }\lfloor \frac{d\epsilon}{2} \rfloor j\mbox{'s }}
p_{i_j}\le \sum_{\mbox{ remaining terms }} p_{i_k}, \mbox{ and }
\label{eqn_gen2}
\end{equation}
\begin{equation}
3\Big{(}\sum_{\mbox {any }\lfloor \frac{d\epsilon}{4} \rfloor
j\mbox{'s }} p_{i_j}\Big{)}\le \sum_{\mbox{ remaining terms }}
p_{i_k} \label{eqn_gen3}
\end{equation}
\end{remark}

\begin{figure}
\centering{\resizebox{3.1in}{1.3in}{\includegraphics{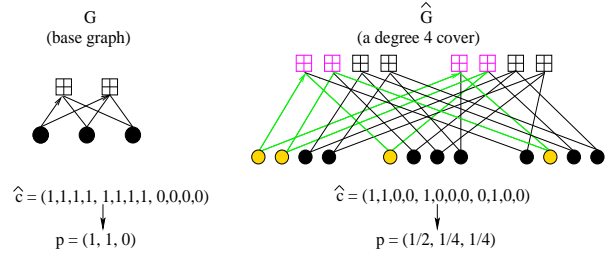}}}
\caption{A pseudocodeword in the base graph (or a valid codeword in a
lift).}
\label{pscw_fig}
\end{figure}

The set of lift-realizable pseudocodewords can also be described
elegantly by means of a polytope, called the fundamental polytope
\cite{ko03p}. In particular, lift-realizable pseudocodewords are
dense in the fundamental polytope. For simple LDPC codes, equations
(1) and (2) are necessary and sufficient conditions for a
pseudocodeword to lie in the fundamental polytope. However, for
generalized LDPC codes, equations (\ref{eqn_zero}), (\ref{eqn_gen}),
(\ref{eqn_gen2}), and (\ref{eqn_gen3}) are necessary but, in
general, not sufficient conditions for a pseudocodeword to lie in
the fundamental polytope.

It was shown in \cite{ko03p,ke05u} that a stopping set in a simple LDPC constraint graph
is the support of a pseudocodeword as defined above. Thus, generalizing the definition
of stopping sets to generalized LDPC code, we have:

\begin{definition}{\rm
A {\em stopping set} in a generalized LDPC constraint graph $G$ is the support of a pseudocodeword {\bf p} of $G$. }
\label{sset_gen}
\end{definition}

Note that this definition of stopping sets for a generalized LDPC code implies the same operational meaning as stopping sets for simple LDPC codes, i.e., the iterative decoder gets stuck if a generalized LDPC code is used for transmission over a BEC and the set of erased positions at the receiver contains a stopping set (as a subset) as defined above.

In Sections 3, 4, 5 and 6, we will consider pseudocodewords and their behavior on the binary symmetric channel (BSC),
and in Section 7, we will consider pseudocodewords on the additive white Gaussian noise (AWGN) channel.
The weight of a pseudocodeword ${\bf p}$ on the BSC is defined as follows \cite{fo01in}.
\begin{definition}{\rm
Let $e$ be the smallest number  such that the sum of the $e$ largest
components of ${\bf p}$ is at least
the sum of the remaining components of ${\bf p}$. Then, the BSC {\em pseudocodeword weight}
of ${\bf p}$ is
\[w_{BSC}({\bf p})=\left\{\begin{array}{cc}
2e,& \mbox{if } \sum_{e \mbox{ largest}} p_i = \sum_{\mbox{remaining
}}p_i\\
2e-1,& \mbox{if } \sum_{e \mbox{ largest}} p_i > \sum_{\mbox{remaining
}}p_i
\end{array} \right .\]}
\label{bsc_pscw_wt}
\end{definition}

\begin{definition}{\rm
The {\em minimum BSC pseudocodeword weight} of an LDPC constraint graph $G$
on the BSC is the minimum weight among all pseudocodewords obtainable from all finite-degree
lifts of $G$. This parameter is denoted by $w^{BSC}_{\min}$.}
\end{definition}
\vspace{0.1in}

\section{Case A}

\vspace{0.1in}
\begin{definition}{\rm Let $0< \alpha <1$ and $0< \delta < 1$.
A $(c,d)$-regular bipartite graph $G$ with $n$ degree $c$ nodes on the left and $m$
degree $d$ nodes on the right is an {\em $(\alpha n,\delta c)$ expander} if for every
subset $U$ of degree $c$ nodes such that $|U|< \alpha n$, the size of the set of
neighbors of $U$, $|\Gamma(U)|$ is at least $\delta c |U|$. }
\label{expander_a}
\end{definition}

Let a $(c,d)$-regular bipartite graph $G$ with $n$ left
vertices  and $m$ right vertices be an $(\alpha n,\delta c)$ expander. An LDPC code is
obtained from $G$ by interpreting the degree $c$ vertices in $G$ as variable nodes
and the degree $d$ vertices as simple parity-check nodes. (See Figure~\ref{caseA}.)

\begin{figure}
\centering{\resizebox{4in}{2.2in}{\includegraphics{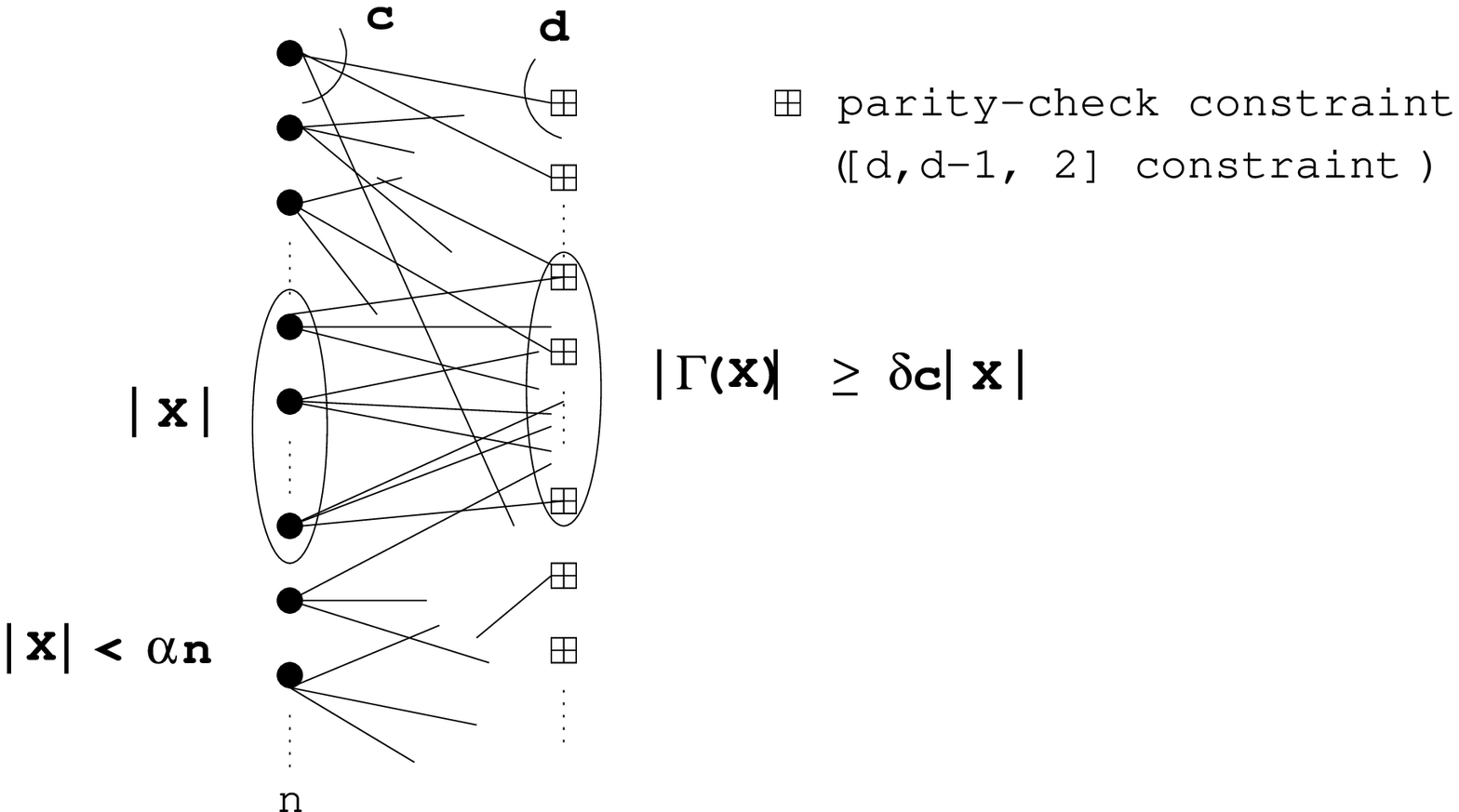}}}
\centerline{\scriptsize Degree $c$ vertices: variable nodes, degree $d$ vertices: simple parity-check constraints.}
\caption{Expander code: Case A.}
\label{caseA}
\end{figure}

\subsection{Minimum distance}
\vspace{0.1in}
\begin{lemma}\cite{si96} If $\delta > 1/2$, the LDPC code obtained from
the $(\alpha n,\delta c)$ expander graph
$G$ as above has minimum distance $d_{\min}\ge \alpha n$.
\label{dmin_A}
\end{lemma}
\vspace{0.1in}

\subsection{Minimum stopping set size}

\begin{lemma} If $\delta > 1/2$, the LDPC code obtained from the $(\alpha n,\delta c)$ expander graph
$G$ as above has a minimum stopping set size  $s_{\min}\ge \alpha n$.
\label{smin_A}
\end{lemma}
\vspace{0.1in}

\begin{proof}

Suppose the contrary that there exists a stopping set $S$ of size smaller than $\alpha n$.  Then by the expansion property
of the graph, the size of the set of
neighbors of $S$ is $|\Gamma(S)|\ge \delta c |S|$. The average number of times a vertex in $\Gamma(S)$ is connected to
the set $S$ is  $\frac{c|S|}{|\Gamma(S)|}\le \frac{c|S|}{\delta c|S|} < 2$. This means that there is at least one vertex in
$\Gamma(S)$ that is connected to the set $S$ only once, contradicting the fact that $S$ is a stopping set.
\end{proof}

Note that the above proof is just an extension of the proof of Lemma \ref{dmin_A} for the lower bound on the minimum distance $d_{\min}$ since it uses
the fact that every check node neighbor of a stopping set is connected to the set at least twice, which is a similar requirement
for a codeword in the proof of Lemma \ref{dmin_A}.

\subsection{Minimum pseudocodeword weight}

\begin{theorem} If $\delta > 2/3+1/3c$ such that $\delta c$ is an
integer, the LDPC code obtained from the $(\alpha n,\delta c)$ expander
graph $G$ as above has a pseudocodeword weight
\[w^{BSC}_{\min}> \frac{2(\alpha n-1)(3\delta-2)}{(2\delta -1)}-1.\]
\label{case_a_thm}
\end{theorem}

\begin{proof}
The proof is by contradiction. Let ${\bf p} = (p_1,\ldots,p_n)$ be
a pseudocodeword in $G$. Without loss of generality, let $p_1\ge
p_2 \ge \ldots \ge p_n$. We will show that if $e$ is the smallest number such that $p_1+p_2+\dots+p_e\ge p_{e+1}+..+p_n$, then $e$ must be more than $\frac{(\alpha n-1)(3\delta-2)}{2\delta-1}$. We will assume a subset $U$ of size $e$ of variable nodes corresponding to the $e$ dominant components of ${\bf p}$ to have a size that is at most $\frac{(\alpha n-1)(3\delta-2)}{2\delta-1}$ and establish the necessary contradiction.

Let $V=\{ v_1,v_2,\dots,v_n\}$ be the set of variable nodes. Let $U = \{v_1,\ldots,v_e\}$ be a set of
$e$ variable nodes corresponding to the $e$
largest components of ${\bf p}$.  Let $\dot{U}=\{v_i \in V| v_i \notin U,
|\Gamma(v_i) \cap \Gamma(U)| \ge (1-\lambda)c + 1\}$, where $\Gamma(X)$ is the set of neighbors of the vertices in
$X$ and $\lambda = 2(1-\delta)+\frac{1}{c}$. Let $U' = U \cup
\dot{U}$. Note that since we assume $\delta c$ to be an integer, $\lambda c$ is also an integer.\\

We want to show that if $|U'| < \alpha n$, then we can find a
set $M$ of edges such that: (i) every node in $U$ is incident with
at least $\delta c$ edges in $M$, (ii) every node in $\dot{U}$ is
incident with at least $\lambda c$ edges in $M$, and (iii) every node in $\Gamma(U')$
is incident with at most one edge in $M$. (Such a set $M$ is
called a {\em $(\delta,\lambda)$-matching} in \cite{feldman}.) Suppose $e = |U| \le \frac{(\alpha n-1)}{(1+\beta)}$,
where $\beta = \frac{(1-\delta)}{(3\delta - 2)}$. Then by Lemma 6 in \cite{feldman}, $|\dot{U}| \le \beta|U|$.
This implies that $|U'| \le (1+\beta)|U| \le (\alpha n - 1)$. Since $G$ is an
$(\alpha n, \delta c)$-expander, this means $|\Gamma(U')| \ge \delta c|U'| = \delta c |U|+\delta c |\dot{U}|$.

We will prove the $(\delta,\lambda)$-matching property in $G$ by
constructing a new bipartite graph $\hat{G}$ as follows. Label the
edges connected to each vertex in $U\cup \dot{U}$ as $\{1,2,\dots,
c\}$. For each vertex $v$ in $U\cup \dot{U}$, create $\delta c$
vertices $v_1,v_2, \dots, v_{\delta c}$ in $\hat{G}$. For every
vertex $w$ in $\Gamma(U')$ in the graph $G$, form a vertex $w$ in
the graph $\hat{G}$. Let $\hat{U}$ correspond to the set of vertices
in $\hat{G}$ that correspond to the copies of vertices in $U\cup
\dot{U}$ and let $W$ be the set of vertices in $\hat{G}$ that
correspond to the vertices in $\Gamma(U')$ in $G$. For a vertex $v$
in $G$, connect the vertex $v_i\in X$, $i=1,2,\dots, \delta c$, to a
node $w\in W$ if and only if the $i^{th}$ edge of $v$ is connected
to $w\in \Gamma(U')$ in $G$. (Note that the $\delta c$ vertices in
$\hat{G}$ that correspond to a node in $U$ correspond to $\delta c$
edges incident on that node in $G$. Furthermore, since $G$ does not
contain multiple edges, each of those $\delta c$ edges are connected
to a distinct node in $\Gamma(U')$, which means that each of the
$\delta c$ copies in $\hat{G}$ corresponding to a node in $U$ are
connected to a distinct node in $W$.) Now, for any subset $X\subset
\hat{U}$, we will always have that $|\Gamma(X)|\ge |X|$ in
$\hat{G}$. This can be seen by the following argument: since the
graph $G$ is an $(\alpha n, \delta c)$ expander and $|U\cup
\dot{U}|< \alpha n$, therefore any subset $Y\subset U\cup \dot{U}$
has the property that $|\Gamma(Y)|\ge \delta c(|Y|)$ in $G$. Choose
$Y$ such that the set of vertices in $X$ in $\hat{G}$ correspond to
the set of vertices in $Y$ in the graph $G$. We have
$|\Gamma(X)|=|\Gamma(Y)|\ge \delta c |Y|\ge \delta c |X|/(\delta c)
= |X|$ since $\Gamma(X)=\Gamma(Y)$ by construction and $\delta
c|Y|\ge |X|$ by construction. Thus, for any subset $X\subset
\hat{U}$ in the graph $\hat{G}$, we have $|\Gamma(X)|\ge |X|$.
Therefore, by Hall's (Marriage) Theorem, there is a matching of all
nodes in $\hat{U}$ (which corresponds to the $\delta c $ copies of
vertices in $U \cup \dot{U}$) to the set in $W$ (or $\Gamma(U')$).
Since $\lambda c < \delta c$ by the choice of $\lambda$, this means that the matching in $\hat{G}$ corresponds to a $(\delta,\lambda)$-matching for the set $U'$ in the graph $G$.\\

Consider all of the check nodes in $\Gamma(U)$ that are incident with
edges from $M$ that are also incident with the vertices in $U$. Let us
call this set of check nodes $T$. We now apply the inequality in
equation (\ref{eqn_basic}) at each of the check nodes in $T$ and combine
the inequalities by summing them. For each check node, the
left-hand side of equation (\ref{eqn_basic}) is chosen to be a component of the pseudocodeword corresponding to a vertex in $U$ if the edge from $M$ that is incident with the check node is also
incident with that vertex in $U$. After combining all such inequalities in all of the above check nodes, we obtain an inequality that has $\delta c(p_1+\dots+p_e)$
on the left hand side since there are at least $\delta c$ edges from each vertex in $U$ that are incident with $M$. Furthermore, by the same argument, there are most $(1-\delta)c$ edges from each vertex in
$U$ that are not in $M$ but are possibly also incident with the above check nodes. Moreover, at most $(1-\lambda)c$ edges from each vertex in $\dot{U}$ are
possibly incident with the above check nodes and at most $(1-\lambda)c$
edges from each vertex in $V\backslash U'$ are possibly incident with the
above check nodes by the definition of $U'$.
Therefore, we have the following inequality when we sum the inequalities obtained above at all the above check nodes:

{\scriptsize \begin{equation}
 \delta c(p_1+\cdots+p_e) \le (1-\delta)c(p_1+\cdots+p_e) + (1-\lambda)c(\sum_{v_i \in \dot{U}} p_i)+(1-\lambda)c(\sum_{v_i \in V\backslash U'} p_i) .\ \
\label{eqn_*1}
\end{equation}}

The above inequality implies that \[p_1+\cdots+p_e \le \frac{(1-\lambda)}{(2\delta - 1)}(p_{e+1}+\cdots+p_n)<p_{e+1}+\cdots+p_n,\]
from the choice of $\lambda$. Thus, the desired contradiction is achieved. From the definition of pseudocodeword weight on the BSC (Definition \ref{bsc_pscw_wt}), we have $w^{BSC}({\bf p}) > 2e-1 = \frac{2(\alpha n - 1)(3\delta - 2)}{(2\delta - 1)}-1$.

\end{proof}
\vspace{0.1in}

\begin{remark}

\begin{itemize}
\item The proof of the above theorem can also be inferred directly by the result in \cite{feldman}. However, we believe the proof presented here is somewhat simpler than the indirect approach in \cite{feldman}.

\item For the case when $\delta =3/4$, the lower bound on the minimum pseudocodeword weight $w_{\min}$ matches the lower bound on $d_{\min}$ and $s_{\min}$ presented in Lemmas \ref{dmin_A} and \ref{smin_A}. This is particularly appealing since an expander code achieving the lower bound on the minimum distance will also achieve the lower bound on the minimum pseudocodeword and will have no pseudocodewords of weight less than the minimum distance.

\end{itemize}
\end{remark}

\section{Case B}

Let a $(c,d)$-regular bipartite graph $G$ with $n$
left  vertices and $m$
right vertices be a $(\alpha n,\delta c)$ expander. (See Definition~\ref{expander_a}.) An LDPC code is
obtained from $G$ by interpreting the degree $c$ vertices in $G$ as variable nodes
and the degree $d$ vertices as sub-code constraints imposed by a $[d,rd,\epsilon d]$ linear block code\footnote{The
parameters of an $[n,k,d]$ binary
linear block code are the block length $n$, the dimension $k$,
and the minimum distance $d$.}. A valid assignment of values to the variable nodes is one where the (binary)
values assigned to the variable nodes connected to each constraint node
satisfy all the constraints imposed by the subcode, meaning that the binary assignments from the variable nodes connected to each constraint
node form a codeword in the subcode. (See Figure~\ref{caseB}.)
Such an LDPC code is called a {\em generalized LDPC code}.
\vspace{0.1in}

\begin{figure}
\centering{\resizebox{4in}{2.2in}{\includegraphics{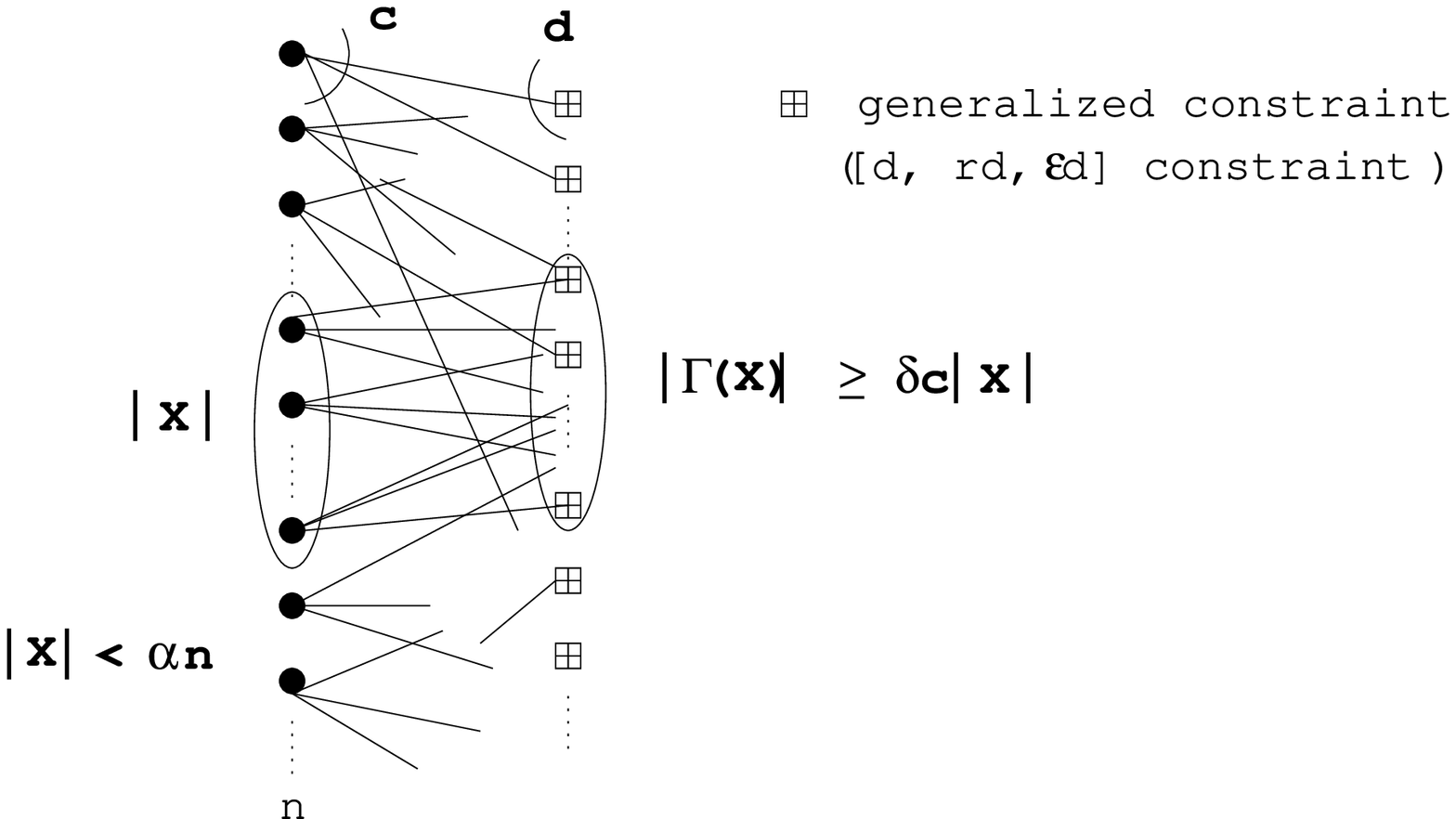}}}
\centerline{\scriptsize Degree $c$ vertices: variable nodes, degree $d$ vertices: sub-code
constraints of a $[d,rd,\epsilon d]$  code.}
\caption{Expander code: Case B.}
\label{caseB}
\end{figure}

\subsection{Minimum distance}

\begin{lemma}\cite{si96} If $\delta > 1/(\epsilon d)$, the LDPC code
obtained from the $(\alpha n,\delta c)$ expander
graph
$G$ as above has minimum distance $d_{\min}\ge \alpha n$.
\end{lemma}
\vspace{0.1in}

\subsection{Minimum stopping set size}
\vspace{0.1in}

A generalized stopping set is as defined in Definition \ref{sset_gen} in Section 2.
Under the assumption that the $[d,rd,\epsilon d]$ subcode has no idle components, meaning that there are no components that are zero in all of the codewords
of the subcode, Definition \ref{sset_gen} reduces to the following: {\em A stopping set in a generalized LDPC code is a set of variable nodes such that every node that is a neighbor of some node $s\in S$ is connected to $S$ at least $\epsilon d$ times.}
\vspace{0.1in}

\begin{lemma} If $\delta > 1/(\epsilon d)$, the LDPC code obtained from the $(\alpha n,\delta c)$ expander graph
$G$ as above has a minimum stopping set size $s_{\min}\ge \alpha n$.
\end{lemma}
\vspace{0.1in}

\begin{proof}

Suppose the contrary that there exists a stopping set $S$ of size smaller than $\alpha n$. Then by Definition \ref{sset_gen},
there is a pseudocodeword ${\bf p}$ whose support has a size smaller than $\alpha n$. By the expansion property
of the graph, the size of the set of
neighbors of $S$ is $|\Gamma(S)|\ge \delta c |S|$. The average number of times a vertex in $\Gamma(S)$ is connected to
the set $S$ is  $\frac{c|S|}{|\Gamma(S)|}\le \frac{c|S|}{\delta c|S|} < d\epsilon$. This means that there is at least one vertex in
$\Gamma(S)$ that is connected to the set $S$ less than $d\epsilon$ times. Therefore, there are
less than $d\epsilon$ non-zero pseudocodeword components connected to that
constraint node in $\Gamma(S)$. If we choose the $d\epsilon/2$ largest components
among them,
then their sum is greater than the sum of the remaining pseudocodeword components at that constraint node. This is a contradiction
to the inequality in Equation \ref{eqn_gen2}, meaning ${\bf p}$ cannot be a
pseudocodeword and therefore $S$ cannot be a stopping set. Thus, the size of $S$
cannot be less than $\alpha n$. \end{proof}

\subsection{Minimum pseudocodeword weight}

\begin{theorem} If $\delta > \frac{2}{(\epsilon d+1)}+\frac{1}{c(\epsilon d+1)}$ such that $\delta c$ is an integer, then the LDPC code obtained from the
$(\alpha n,\delta c)$ expander graph $G$ has a minimum pseudocodeword weight
\[w^{BSC}_{\min}> \frac{2(\alpha n-1)((d\epsilon+1)\delta-2)}{(d\epsilon\delta -1)} -1.\]
\label{case_b_thm}
\end{theorem}
\vspace{0.1in}

\begin{proof} The proof is by contradiction. Suppose $G$ is an $(\alpha n, \delta c)$-expander, where $\delta > \frac{2}{(d\epsilon +
1)}+\frac{1}{(d\epsilon+1)c}$. Then, assuming ${\bf p}$ is
a pseudocodeword of the LDPC constraint graph $G$, the proof follows that of
Case A by choosing a set of variable nodes $U$ corresponding to the $e$ dominant components of the pseudocodeword and
letting $|U|=e\le \frac{(\alpha n-1)}{1+\beta}$, where $\beta=\frac{1-\delta}{(d\epsilon+1)\delta-2}$.
We need to show that
$w_{\min}^{BSC} > 2e-1 = 2\frac{(\alpha n - 1)((d\epsilon+1)\delta - 2)}{(d\epsilon\delta - 1)}-1$. By
using a strong subcode, the $\delta$ required is less than that in Case A, thereby allowing $\alpha$ to be
larger and yielding a larger bound overall. The argument is the same
as in the proof of Case A, where now we set $\lambda = 2-d\epsilon \delta + \frac{1}{c}$.

Following the proof of Theorem 1, we will first show that if
$|U|=e\le \frac{(\alpha n -1)}{1+\beta}$, then $|\dot{U}|\le \beta |U|$.
Suppose to the contrary, $|\dot{U}|> \beta |U|$, then that means there is some subset $\ddot{U}\subset \dot{U}$ such that
$|\ddot{U}|=\lfloor \beta |U|\rfloor +1$. This means that $|U\cup \ddot{U}|=|U|+\lfloor \beta |U|\rfloor +1\le (1+\beta)|U|+1 = \alpha n$.
Since $G$ is an $(\alpha n, \delta c)$ expander, we then have $|\Gamma(U\cup \ddot{U})|\ge \delta c(|U|+|\ddot{U}|)$.
However, observe that $|\Gamma(U\cup \ddot{U})|=|\Gamma(U)|+|\Gamma(\ddot{U})\backslash \Gamma(U)|\le c|U|+ (\lambda c-1)|\ddot{U}|$
since $|\Gamma(U)|\le c|U|$ and $|\Gamma(\ddot{U})\backslash \Gamma(U)|\le (\lambda c-1)|\ddot{U}|$ by definition. Combining the above
inequalities, we have $\delta c(|U|+|\ddot{U}|) \le c|U|+(\lambda c-1)|\ddot{U}|$, implying $|\ddot{U}|\le \frac{c(1-\delta)|U|}{c(\delta-\lambda)+1} = \beta |U|$. This contradicts the
choice of $\ddot{U}$ above. Thus, if $|U|=e\le \frac{(\alpha n-1)}{1+\beta}$, then
$|\dot{U}|\le \beta |U|$.

Following the rest of the proof of Theorem 1 and using the inequality in Equation \ref{eqn_gen} for the pseudocodeword components,
the first inequality in equation (\ref{eqn_*1}) in the proof of Case A now becomes
\vspace{-0in}
{ \[ (d\epsilon - 1)\delta c(p_1+\cdots+p_e) \le
(1-\delta)c(p_1+\cdots+p_e)\vspace{-0.1in}\]\[\vspace{-0.1in}+(1-\lambda)c(\sum_{v_i
\in \dot{U}}
p_i)+(1-\lambda)c(\sum_{v_i \in V\backslash U'} p_i). \]}

This yields
\[p_1+\cdots+p_e \le \frac{(1-\lambda)}{(d\epsilon \delta - 1)}(p_{e+1}+\cdots+p_n)<p_{e+1}+\cdots+p_n\]
from the choice of $\lambda$. Thus, by Definition \ref{bsc_pscw_wt}, the weight of ${\bf p}$ is
$w({\bf p})> 2e-1 $.
\end{proof}
\vspace{0.1in}

\begin{remark}

\begin{itemize}
\item Since $d\epsilon \ge 2$ for any judicious choice of subcode, the lower bound in Theorem 2 is always greater than the lower bound
in Theorem 1. Further, the graph need not be as good an expander in Case B as in Case A for the lower bound to hold. Thus, using
strong subcodes is advantageous for constructing good LDPC codes from expander graphs.

\item Note that for $\delta=\frac{3}{d\epsilon +2}$,  the lower bound on the
pseudocodeword
weight equals the lower bound on the minimum distance and minimum stopping set size.

\end{itemize}
\end{remark}

\section{Case C}

\vspace{0.1in}

\begin{definition}{\rm
A connected, simple, graph $G$ is said to be a {\em $(n,d,\mu)$ expander} if $G$ has $n$ vertices,
is $d$-regular, and the second largest eigenvalue of $G$ (in absolute value) is $\mu$.}
\end{definition}

Let a $d$-regular graph $G$ be an $(n,d,\mu)$
expander. An LDPC code is obtained from $G$ by interpreting the edges in $G$ as variable nodes
and the degree $d$ vertices as constraint nodes imposing
constraints of an $[d,rd,\epsilon d]$ linear block code. (See Figure~\ref{caseC}.)
The resulting LDPC code has block length $N=nd/2$ and rate $R\ge 2r-1$.

\begin{figure}
\centering{\resizebox{5.75in}{2.9in}{\includegraphics{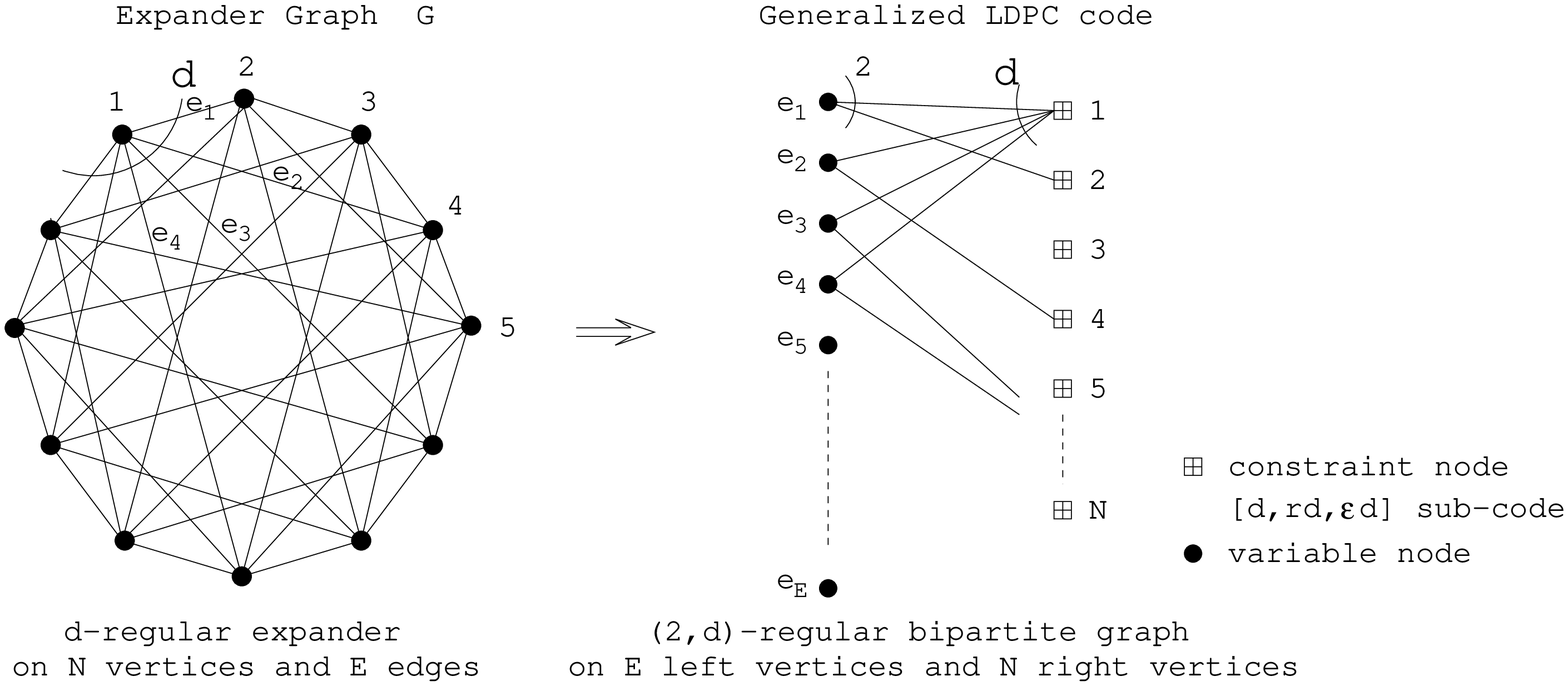}}}
\centerline{\scriptsize Edges: variable nodes, degree $d$ vertices: constraints of a $[d,rd,\epsilon d]$ code.}
\caption{Expander code: Case C.}
\label{caseC}
\end{figure}

\vspace{0.1in}

We now state a particularly useful result by Alon and Chung
\cite{alon_chung,si96}
describing the expansion of a $d$-regular graph.

\begin{lemma}(Alon-Chung)
Let $G$ be a $d$-regular graph on $n$ vertices and let $\mu$ be the second largest eigenvalue of its
adjacency matrix. Then every subset $S$ of $\gamma n$ vertices contains at most $\frac{nd}{2}(\gamma^2 +\frac{\mu}{d}(\gamma-\gamma^2))$
edges in the subgraph induced by $S$ in $G$.
\label{alon_chung_lemma}
\end{lemma}

\subsection{Minimum distance}
\begin{lemma}\cite{si96}  The LDPC code obtained from an $(n,d,\mu)$
expander graph $G$ as above has minimum distance $d_{\min}\ge N \frac{(\epsilon -\frac{\mu}{d})^2}{(1-\frac{\mu}{d})^2}$.
\end{lemma}
\vspace{0.1in}

Note that the above result of Sipser and Spielman can be improved by a tighter bound in the last step of their proof
in \cite{si96} to $d_{\min}\ge N \epsilon \frac{(\epsilon -\frac{\mu}{d})}{(1-\frac{\mu}{d})}$.

\subsection{Minimum stopping set size}

\vspace{0.1in}\begin{lemma}  The LDPC code obtained from an $(n,d,\mu)$ expander graph
$G$ has a minimum stopping set size $s_{\min}\ge N \epsilon \frac{(\epsilon -\frac{\mu}{d})}{(1-\frac{\mu}{d})}$.
\label{case_c_lemma}
\end{lemma}
\vspace{0.1in}
Note that we again use Definition \ref{sset_gen} for stopping sets in $G$.

\begin{proof}

Let $S$ be a subset of variable nodes (edges in $G$) of size
$\frac{nd}{2}(\gamma^2+\frac{\mu}{d}(\gamma-\gamma^2))$ representing a
stopping set in $G$. Then $S$ is the support of some pseudocodeword ${\bf
p}$ in $G$. By the Alon-Chung lemma, the set $S$ has at least $\gamma n$
constraint node neighbors $\Gamma(S)$. Since each edge in $S$ has two
constraint node neighbors in $\Gamma(S)$, this implies that the average
number of edges in $S$ connected to a constraint node in $\Gamma(S)$ is
at most $\frac{2 \frac{nd}{2}(\gamma^2+\frac{\mu}{d}(\gamma-\gamma^2))}{\gamma
n}$. However if
\begin{equation}
\frac{2
\frac{nd}{2}(\gamma^2+\frac{\mu}{d}(\gamma-\gamma^2))}{\gamma
n}=d(\gamma+\frac{\mu}{d}(1-\gamma)) < \epsilon d, \ \
\label{eqn_**2}
\end{equation}
then there is
at least one node in $\Gamma(S)$ that is connected to $S$ fewer than
$\epsilon d$ times to $S$. That means that fewer than $\epsilon d$
non-zero components of ${\bf p}$ are connected to a constraint node. It
can now be shown that the inequality in Equation \ref{eqn_gen2} is
violated, implying that ${\bf p}$ cannot be a pseudocodeword (and, $S$ is
not a stopping set.)

The above inequality in equation (\ref{eqn_**2}) holds for $\gamma <
\frac{\epsilon-\frac{\mu}{d}}{1 - \frac{\mu}{d}}$.
Substituting the value of $\gamma$ in $|S|=\frac{nd}{2}(\gamma^2+\frac{\mu}{d}(\gamma-\gamma^2))$,
we infer that the graph cannot contain a stopping set of size less than $\frac{nd}{2}\epsilon\frac{(\epsilon -\frac{\mu}{d})}{(1-\frac{\mu}{d})}$. Hence,
\[s_{\min}\ge \frac{nd}{2}\epsilon\frac{(\epsilon -\frac{\mu}{d})}{(1-\frac{\mu}{d})} =N\epsilon\frac{(\epsilon -\frac{\mu}{d})}{(1-\frac{\mu}{d})}.\]
\end{proof}

\subsection{Minimum pseudocodeword weight}

\begin{theorem}  The LDPC code obtained from an $(n,d,\mu)$ expander graph
$G$  has a minimum BSC pseudocodeword weight lower bounded as follows: \[w^{BSC}_{\min}\ge N \epsilon \frac{(\frac{\epsilon}{2}-\frac{\mu}{d})}{(1-\frac{\mu}{d})} .\]
\label{case_c_thm}
\end{theorem}
\vspace{0.1in}

\begin{proof}
For sake of simplicity, we assume that $\frac{d\epsilon}{4}$ is an
integer. However, it is easy to extend the proof for any value of
$d\epsilon$. The $d$-regular graph $G$ can be transformed to a
$(2,d)$-regular bipartite graph $G'$ by representing every edge in
$G$ by a vertex in $G'$ ({\em the variable nodes}) and every vertex in $G$
by a vertex in $G'$ ({\em the constraint nodes})
and connecting the variable nodes to the constraint nodes in $G'$ in a
natural way. The variable nodes have degree two and they represent
codebits of the LDPC code $C$, whereas the constraint nodes have
degree $d$ and each represents a $[d,rd,\epsilon d]$-subcode
constraints.

Let ${\bf p}=(p_1,p_2,\dots,p_N)$ be a pseudocodeword, where $N=\frac{nd}{2}$ is the number of edges in $G$ and also the length of the LDPC code. Without loss of generality, let us assume that $p_1 \ge p_2 \ge \cdots \ge p_N$. Let $e$ be the smallest number such that $p_1+p_2\dots+p_e > p_{e+1}+\dots+p_N$.
Let $X_e$ be the set of edges in $G$ that correspond to the support of the $e$
largest components of ${\bf p}$, and let $\Gamma(X_e)$ be the set of
vertices incident on $X_e$. Note that in the transformed graph $G'$, $X_e$
is a subset of the variable nodes, and $\Gamma(X_e)$ is its set of
neighbors.

Let $|X_e|= \frac{nd}{2}(\gamma^2+\frac{\mu}{d}(\gamma-\gamma^2))$,
where $\gamma \le
(\frac{\frac{\epsilon}{2}-\frac{\mu}{d}}{1-\frac{\mu}{d}})$. Since
$G$ is an $(n,d,\mu)$ graph, we have $|\Gamma(X_e)|\ge \gamma n$. We
now claim that there is a set of edges $M$ in $G'$ called an
$\epsilon$-matching such that (i) every vertex in $X_e$ in the graph
$G'$ is incident with at least one edge from $M$ and (ii) every
vertex in $\Gamma(X_e)$ in the graph $G'$ is incident with at most
$d\epsilon/4$ edges from $M$.

Given the claim, we can apply the pseudocodeword inequality from
equation \ref{eqn_gen3} at each of the vertices in $\Gamma(X_e)$
that is incident with edges from $M$. For each such vertex $w$ in
$\Gamma(X_e)$, the left-hand side of equation (\ref{eqn_gen3}) is
chosen to have the $d\epsilon/4$ or less components of the
pseudocodeword that correspond to the vertices in $X_e$ that are
connected to $w$ via an edge from $M$. After combining all such
inequalities in all of the above constraint nodes, we obtain an
inequality that has $3(p_1+\dots+p_e)$ on the left hand side.

Furthermore, there is at most one edge from each vertex in $X_e$
that is not in $M$ but is possibly also incident with the above
constraint nodes in $\Gamma(X_e)$. Moreover, at most two edges from
each vertex in $V-X_e$ are possibly incident with the above
constraint nodes. Therefore, after applying the pseudocodeword
inequality (equation \ref{eqn_gen3}) as above at each of these
constraint nodes and summing these inequalities, we obtain the
following inequality:
\[3\Big{(}\sum_{i\in X_e} p_i\Big{)} \le \sum_{i\in X_e} p_i + 2\Big{(}\sum_{i \notin X_e} p_i\Big{)}.\]
Simplifying, we get
\[\Big{(}\sum_{i\in X_e} p_i\Big{)} \le \Big{(}\sum_{i \notin X_e} p_i\Big{)}.\]

By the definition of the pseudocodeword weight on the BSC channel (see Definition
\ref{bsc_pscw_wt}), we have that the pseudocodeword weight of ${\bf p}$ is
$w({\bf p}) \ge 2|X_e|$. Since
$|X_e|=\frac{nd}{2}(\gamma^2+\frac{\mu}{d}(\gamma-\gamma^2))$, for $\gamma
\le (\frac{\frac{\epsilon}{2}-\frac{\mu}{d}}{1-\frac{\mu}{d}})$, we have
$w({\bf p})\ge 2
(\frac{nd}{2})(\frac{\epsilon}{2})(\frac{\frac{\epsilon}{2}-\frac{\mu}{d}}{1-\frac{\mu}{d}})=N\epsilon(\frac{\frac{\epsilon}{2}-\frac{\mu}{d}}{1-\frac{\mu}{d}})$.
This proves the desired lower bound on $w_{\min}$.

To prove the claim, observe that for any set $X$ of left vertices in
$G'$ such that $|X|=N (\gamma^2+\frac{\mu}{d}(\gamma-\gamma^2))$
where $\gamma\le
(\frac{\frac{\epsilon}{2}-\frac{\mu}{d}}{1-\frac{\mu}{d}})$, we have
$|\Gamma(X)| \ge \gamma n \ge \frac{4}{d\epsilon}|X|$. In other
words, for every $X$ such that $|X|=N
(\gamma^2+\frac{\mu}{d}(\gamma-\gamma^2))$ where $\gamma\le
(\frac{\frac{\epsilon}{2}-\frac{\mu}{d}}{1-\frac{\mu}{d}})$, we have
$\frac{d\epsilon}{4}|\Gamma(X)|\ge |X|. $ Thus, for any $X\subseteq
X_e$, we have
\begin{equation}
(\frac{d\epsilon}{4}) |\Gamma(X)|\ge |X|.
\label{eqn_caseC_new2}
\end{equation}

We now prove the claim that there is an $\epsilon$-matching for the
set $X_e$ using contradiction. The proof is very similar to the
converse of Hall's marriage theorem. We want to show that there is a
set of edges $M$ in $G'$ such that: (i) every $v\in X_e$ is incident
with at least one edge from $M$ and (ii) every $w\in \Gamma(X_e)$ is
incident with at most $d\epsilon/4$ edges from $M$. We will prove
this by showing that there is a matching $M$ such that: (i) every
$v\in X_e$ is incident with exactly one edge from $M$ and (ii) every
$w\in \Gamma(X_e)$ is incident with either exactly $d\epsilon/4$
edges from $M$ or zero edges from $M$.

We consider the induced subgraph $G''$ of $X_e$ in $G'$. Suppose to
the contrary no such matching exists, then we will assume that there
is a maximum matching $M'$ such that the maximum number of vertices
in $X_e$ are matched to the vertices in $\Gamma(X_e)$ as described
above. That is, $M'$ is the maximum matching such that as many
vertices in $X_e$ are each incident with one edge from $M'$ and all
the vertices in $\Gamma(X_e)$ are incident with either zero or
exactly $d\epsilon/4$ edges from $M'$. Since we assume that not all
the vertices in $X_e$ are incident with edges in $M'$, there is a
vertex $v\in X_e$ that is not incident with any edge from $M'$. Now,
we let $S$ be the set of vertices in $X_e$ that are connected to $v$
by an $M'$-alternating path\footnote{Refer to \cite{vanLint} for the
definition of an $M'$-alternating path.} in $G''$ and let $T$ be the
set of vertices in $\Gamma(X_e)$ that are connected to $v$ by an
$M'$-alternating path in $G''$. Then, it is clear that $S\subset
X_e$ and $\Gamma(S) =T$. Furthermore, every vertex in $S-v$ has one
edge incident from $M'$ that is connected to some vertex in $T$ and
every vertex in $T$ has $d\epsilon/4$ edges incident from $M'$ that
are connected to vertices in $S-v$. (Since $M'$ is a
maximum-matching, it is easy to show that there is no
$M'$-augmenting path as defined in \cite{vanLint}.) This means that
$(\frac{d\epsilon}{4})|T|=|S|-1$, which contradicts equation
\ref{eqn_caseC_new2}. This proves that there exists a matching $M$
as described above.
\end{proof}
\vspace{0.1in}

The above proof holds even when $d\epsilon/4$ is not an integer. In
that case we simply replace $\frac{d\epsilon}{4}$ with $\lfloor
\frac{d\epsilon}{4}\rfloor$ in the above when $d\epsilon/4>1$. In
the case when $d\epsilon/4<1$, the proof is trivial since the
$\epsilon$-matching condition follows directly from Hall's marriage
theorem.

\begin{remark}
\begin{itemize}
\item Note that the lower bound on the minimum pseudocodeword weight closely resembles  the lower bound on the minimum distance and the minimum stopping set size. The only difference is a factor of two in the $\epsilon$ term within the braces in Lemma \ref{case_c_lemma} and Theorem \ref{case_c_thm}.
\item The lower bound suggests that if one were to use good expanding graphs such as the Ramanujan graphs from the construction
in \cite{lu88a} and choose an appropriate choice of subcodes having minimum distance at least twice the second eigenvalue of the expander
then the resulting code will have a good pseudocodeword weight and a good minimum distance. This is interesting for designing
codes that are good for iterative decoding or LP decoding.
\end{itemize}
\end{remark}

\section{Case D}
\vspace{0.1in}
\begin{definition}{\rm A $(c,d)$-regular bipartite graph $G$ on $m$ left vertices and $n$ right vertices is a
{\em $(c,d,m,n,\mu)$ expander} if the second largest eigenvalue of $G$ (in absolute value) is $\mu$.}
\end{definition}

 Let a $(c,d)$-regular bipartite graph $G$ be an
$(c,d,m,n,\mu)$ expander.
An LDPC code is obtained from $G$ by interpreting the edges in $G$ as variable nodes,
the degree $c$ left vertices as sub-code constraints imposed by an $[c,r_1 c,\epsilon_1 c]$ linear block code, and
the degree $d$ vertices as constraint nodes imposing
constraints of an $[d,r_2d,\epsilon_2 d]$ linear block code. (See
Figure~\ref{caseD}.)
The resulting LDPC code has block length $N=mc=nd$ and rate $R\ge r_1+r_2-1$.

\begin{figure}
\centering{\resizebox{4in}{2.2in}{\includegraphics{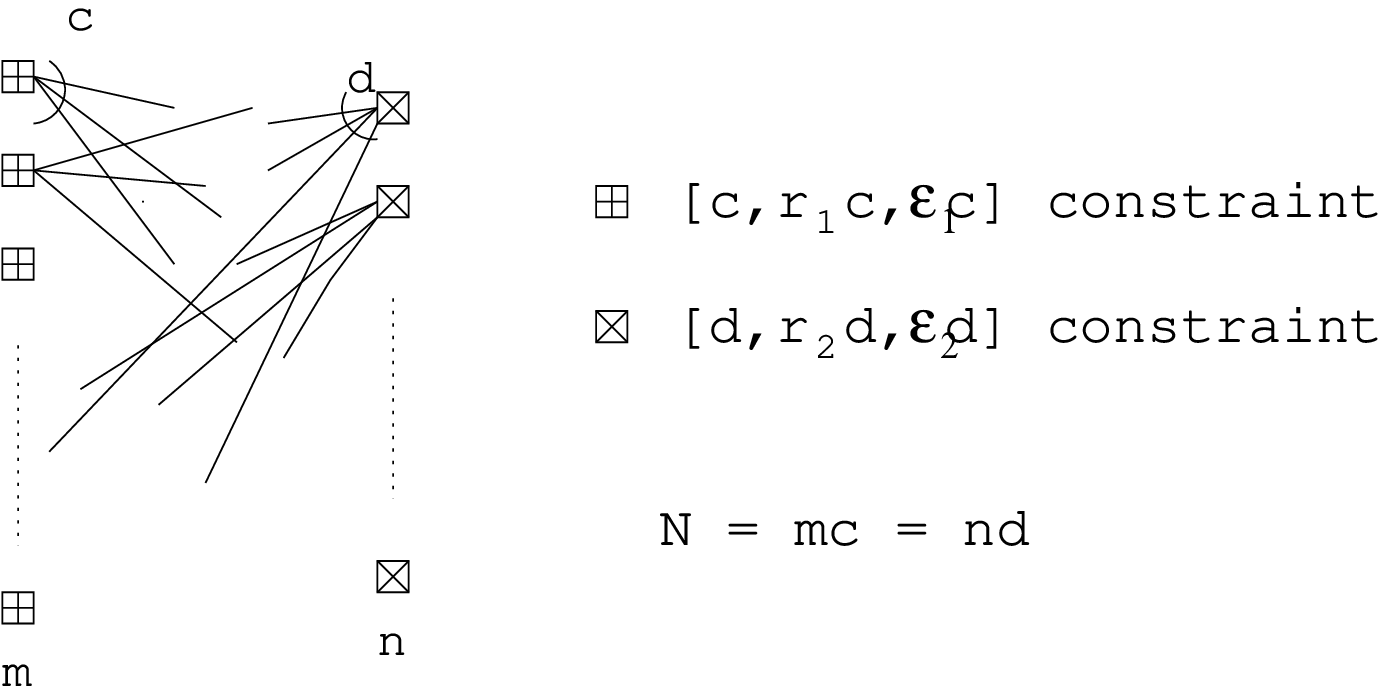}}}
\centerline{\scriptsize Edges: variable nodes, degree $c$ vertices: $[c,r_1c,\epsilon_1c]$ constraints, degree $d$ vertices: $[d,r_2d,\epsilon_2d]$ constraints.}
\caption{Expander code: Case D.}
\label{caseD}
\end{figure}

\vspace{0.1in}

We state a useful result by Janwa and Lal \cite{ja03} describing the edge-expansion of a regular bipartite graph $G$.
\begin{lemma}(Janwa-Lal, edge-expansion)
Let $G$ be a $(c,d)$-regular bipartite graph on $m$ vertices on the left and $n$ vertices on the right and let $\mu$ be its second largest eigenvalue.
If $S$ and $T$ are two subsets of the left and the right vertices, respectively, of $G$, then
the number of edges in the induced sub-graph of $S$ and $T$ in $G$ is at most
\[|E(S,T)|\le \frac{d}{m}|S||T|+\frac{\mu}{2}(|S|+|T|).\]
\label{janwa_lal_lemma}
\end{lemma}

\subsection{Minimum distance}

\begin{lemma}\cite{ja03}  If $\epsilon_2 d\ge \epsilon_1 c>\mu/2$, the
LDPC code obtained from the $(c,d,m,n,\mu)$ expander graph
$G$ as above has minimum distance \[d_{\min}\ge N\left( \epsilon_1 \epsilon_2
-\frac{\mu}{2\sqrt{cd}}(\epsilon_1\sqrt{\frac{c}{d}}+\epsilon_2\sqrt{\frac{d}{c}})\right) .\]
\end{lemma}
\vspace{0.1in}

\subsection{Minimum stopping set size}
\vspace{0.1in}

We again use the generalized definition of stopping set in Definition \ref{sset_gen}. Under the assumption that the
$[d,r_2d,\epsilon_2 d]$ and $[c,r_1c,\epsilon_1 c]$ subcodes have no idle components, meaning that there are no components that are zero in all of the codewords
of either of the subcodes, Definition \ref{sset_gen} reduces to the following: {\em
A stopping set in a generalized LDPC code as in Case D is a set of variable nodes such that every
node that is a degree $c$ neighbor of some node $s\in S$ is connected to $S$ at least $\epsilon_1 c$ times
and every node that is a degree $d$ neighbor of some node $s\in S$ is connected to $S$ at least $\epsilon_2 d$ times.}
\vspace{0.1in}

\begin{lemma} The LDPC code
obtained from the $(c,d,m,n,\mu)$ expander graph
$G$ has a minimum stopping set size \[s_{\min}\ge N\left( \epsilon_1 \epsilon_2
-\frac{\mu}{2\sqrt{cd}}(\epsilon_1\sqrt{\frac{c}{d}}+\epsilon_2\sqrt{\frac{d}{c}})\right) .\]
\end{lemma}
\vspace{0.1in}
Note that when $\min\{\epsilon_2 d, \epsilon_1 c\}> \mu$, the lower bound in the above is positive and meaningful.
\begin{proof}
Let $X$ be a stopping set corresponding to a subset of edges in $G$ and let $S$ and $T$ be the set of left and right neighbors, respectively,
of $X$ in $G$. Then $X$ is the support of some pseudocodeword ${\bf p}$ in $G$. Suppose there is some node in $S$
that is connected fewer that $c\epsilon_1$ times to the edges in $X$, then the inequality in Equation \ref{eqn_gen2} is violated
by the pseudocodeword components at that constraint node.
Similarly, if some node in $T$ is connected fewer than $d\epsilon_2$ times
to the edges in $X$, then the corresponding pseudocodeword components will not satisfy all the inequalities in Equation \ref{eqn_gen2}.
Thus, every node in $S$ is connected to $X$ at least $c\epsilon_1$
times and every node in $T$ is connected to $X$ at least $d\epsilon_2$ times. This means
$|S|\le \frac{|X|}{c\epsilon_1}$ and $|T|\le \frac{|X|}{d\epsilon_2}$. By Lemma \ref{janwa_lal_lemma},
we have \[ |X|\le |E(S,T)|\le \frac{d}{m}|S||T|+\frac{\mu}{2}(|S|+|T|).\]
This can be further bounded
as \[|X| \le \frac{d}{m}|S||T|+\frac{\mu}{2}(|S|+|T|)\le \frac{d}{m}\frac{|X|^2}{cd\epsilon_1\epsilon_2}+\frac{\mu}{2}(\frac{1}{c\epsilon_1}+\frac{1}{d\epsilon_2})|X|.\]
Simplifying, we obtain \[|X|\ge  mc\Big{(}\epsilon_1\epsilon_2-\frac{\mu}{2cd}(\epsilon_1c+\epsilon_2 d)\Big{)} = N\Big{(}\epsilon_1\epsilon_2-\frac{\mu}{2\sqrt{cd}}(\epsilon_1\sqrt{\frac{c}{d}}+\epsilon_2\sqrt{\frac{d}{c}})\Big{)}. \]
\end{proof}

\subsection{Minimum pseudocodeword weight}

\begin{theorem} If $\epsilon_2 d\ge \epsilon_1 c$, the LDPC code
obtained from the $(c,d,m,n,\mu)$ expander graph
$G$  has a minimum pseudocodeword weight \[w^{BSC}_{\min}\ge
N\frac{c}{d}\epsilon_1(\frac{\epsilon_1}{2}-\frac{\mu}{c}).\]
\label{case_d_thm}
\end{theorem}
Note that the above lower bound is positive and meaningful when $\epsilon_1 c>2\mu$.

\begin{proof}
Let ${\bf p}=(p_1,p_2,\dots,p_N)$ be a pseudocodeword. Without loss of generality,
let us assume that $p_1 \ge p_2 \ge \cdots \ge p_N$. Let $e$ be the smallest
number such that
\begin{equation}
p_1+p_2\dots+p_e\ge p_{e+1}+\dots+p_N. \ \
\label{eqn_***3}
\end{equation}

Let $X_e$ be the set of edges in $G$ that correspond to the support of
the $e$ largest components of ${\bf p}$. Now we define a set $S$ as the
set of left neighbors (degree $c$ neighbors) to the edges in $X_e$, and
similarly define a set $T$ as the set of right neighbors (degree $d$
neighbors) to $X_e$. The $(c,d)$-regular graph $G$ can be transformed to
a graph $G'$ by representing every edge in $G$ by a vertex (called a
left-vertex) in $G'$, every vertex of degree c in $G$ by a vertex (called
a right-left vertex) in $G'$, every vertex of degree $d$ in $G$ by a
vertex (called a right-right vertex) in $G'$ and by connecting the edges
from the left vertices to the right-left and right-right vertices in $G'$
in a natural way. The left vertices have degree two and they represent
variable nodes of the LDPC code $C$, whereas the right-left vertices have
degree $c$ and represent $[c,r_1c,\epsilon_1 c]$-subcode constraints and
the right-right vertices have degree $d$ and represent
$[d,r_2d,\epsilon_2 d]$-subcode constraints. Note that $\Gamma(X_e)=S\cup
T$ in $G'$.

Let $|X_e|\le
N\frac{c}{2d}\epsilon_1(\frac{\epsilon_1}{2}-\frac{\mu}{c}).$ Now
let us consider two cases.

\underline{Case 1:} Suppose $|S\cup T|=|S|+|T| <
\frac{4}{c\epsilon_1}|X_e|$. Then, Since $G$ is a $(c,d,m,n,\mu)$
graph, we have
\[|X_e|\le \frac{d}{m}|S||T|+\frac{\mu}{2}(|S|+|T|)\]
Note that $|S||T|\le \frac{(|S|+|T|)^2}{4}$.  Hence, we have
\[|X_e|< \frac{d}{m}
\frac{16|X_e|^2}{4c^2\epsilon_1^2}+\frac{\mu}{2}\frac{4|X_e|}{c\epsilon_1}\]
On simplifying, the above yields
\[ |X_e|> N\frac{c}{2d}\epsilon_1(\frac{\epsilon_1}{2}-\frac{\mu}{c}).\]
This inequality contradicts the assumption on the size of $X_e$.

\underline{Case 2:} Suppose $|S\cup T| \ge
\frac{4}{c\epsilon_1}|X_e| $. Then we claim that there is a set of
edges $M$ in $G'$ called an $\epsilon_1$-matching such that (i)
every vertex in $X_e$ in the graph $G'$ is incident with at least
one edge from $M$ and (ii) every vertex in $S\cup T$ in the graph
$G'$ is incident with at most $c\epsilon_1/4$ edges from $M$.

The rest of the proof is similar to that for Theorem 3. Given the
claim, it is easy to show that by applying the pseudocodeword
inequality from equation \ref{eqn_gen3} at all the nodes in $S\cup
T$ that are incident with edges from $M$ and summing them, we can
arrive at at an inequality of the form
\[\sum_{i\in X_e} p_i \le \sum_{i \notin X_e} p_i.\] This will
prove that the weight of ${\bf p}$ is $w({\bf p})\ge 2|X_e|$
implying that the minimum pseudocodeword weight is
\[w_{\min}^{BSC}\ge 2|X_e| \] Hence
\[w_{\min}^{BSC}\ge
2N\frac{c}{2d}\epsilon_1(\frac{\epsilon_1}{2}-\frac{\mu}{c})=N\frac{c}{d}\epsilon_1(\frac{\epsilon_1}{2}-\frac{\mu}{c}).\]

The proof for the matching also follows the same arguments as that
in Theorem 3. We derive the condition for proving the matching as
follows: For any subset $X\subseteq X_e$, let $S_X$ be the set of
right-left neighbors of $X$ in $G'$ and let $T_X$ be the set of
right-right neighbors of $X$ in $G'$. Note that $|S_X\cup
T_X|=|S_X|+|T_X| \ge \frac{4}{c\epsilon_1}|X|$. Otherwise, using the
argument in case 1 and the fact that $G$ is a $(c,d,m,n,\mu)$
expander, it can be shown that $|X|>
N\frac{c}{2d}\epsilon_1(\frac{\epsilon_1}{2}-\frac{\mu}{c}) \ge
|X_e|$, which is a contradiction. Thus, for any subset $X\subseteq
X_e$, we have
\begin{equation}
|S_X\cup T_X|\ge \frac{4}{c\epsilon_1}|X|. \label{eqn_CaseD_new}
\end{equation}
The rest of the proof is similar to that in Theorem 3.
\end{proof}

\vspace{0.1in}

\begin{remark}

\begin{itemize}
\item Note that  if $c\epsilon_1\ge d\epsilon_2\ge
2\mu$, then it can be shown using a similar proof as in Theorem 4 that
$w_{\min}^{BSC}\ge N\frac{d}{c}\epsilon_2(\frac{\epsilon_2}{2}-\frac{\mu}{d})$.
\item  Observe that the lower bound on the minimum pseudocodeword weight
is slightly weaker compared to the lower bound on the
minimum distance and the minimum stopping set size, since the proof in Theorem 4 exploits the strength of only one the
subcodes -- namely, the subcode with the smaller distance. We however believe that this can be improved to give a much stronger
result as stated below.
\item Note that in the case where $c=d$, $m=n$, and $\epsilon_1=\epsilon_2=\epsilon$, the result in Theorem 4 closely resembles
the result in Theorem 3 and is almost equal to the lower bounds on the minimum distance and the minimum stopping set size.
\item The lower bound suggests that if one were to use good expanding graphs such as the bipartite Ramanujan graphs from the construction
in \cite{lu88a} and choose an appropriate choice of subcodes having minimum distance at least twice the second eigenvalue of the expander
then the resulting code will have a good pseudocodeword weight and a good minimum distance. Once again, this is interesting for designing
codes that are good for iterative decoding or LP decoding. Furthermore, with different choices of $c$ and $d$, there is greater
flexibility in the designing good codes using the construction in Case D than that in Case C.
\end{itemize}
\end{remark}

We believe that Theorem 4 can be improved to a stronger result as follows:
\begin{conjecture}
 If $\epsilon_2 d\ge \epsilon_1 c>2\mu$, the LDPC code
obtained from the $(c,d,m,n,\mu)$ expander graph
$G$  has a minimum pseudocodeword weight \[w^{BSC}_{\min}\ge
N(\frac{\epsilon_1\epsilon_2}{2}-\frac{\mu}{2\sqrt{cd}}(\epsilon_1\sqrt{\frac{c}{d}}+\epsilon_2\sqrt{\frac{d}{c}})).\]
\end{conjecture}

\section{A parity-oriented lower bound}
\begin{definition}
The weight of a pseudocodeword ${\bf q}=(q_1,q_2,\dots,q_n)$ of an LDPC constraint graph $G$ on the AWGN channel is defined as \cite{fo01in,wi96t} \[w^{AWGN}({\bf q})=\frac{(\sum_{i=1}^n q_i)^2}{(\sum_{i=1}^n q_i^2)}.\]
\end{definition}
\vspace{0.1in}

The following bound on the minimum pseudocodeword weight on the AWGN
channel is an adaptation of Tanner's parity-oriented lower bound on the
minimum distance \cite{ta01}. Further, this bound complements the
bit-oriented bound obtained by Vontobel and Koetter \cite{vo04} which is
also a lower bound on the minimum pseudocodeword weight in terms of the
eigenvalues of the adjacency matrix of $G$, obtained using a slightly
different argument. \vspace{0.1in}

\begin{theorem} Let $G$ be a connected $(j,m)$-regular bipartite graph representing an LDPC code with an $r\times n$ parity check matrix $H$.
Then the minimum pseudocodeword weight on the AWGN channel is lower
bounded as
\[ w_{\min}^{AWGN} \ge \frac{n(4j - \mu_2 m)}{(\mu_1 - \mu_2)m} , \]
where $\mu_1=jm$ and $\mu_2$ are the largest and second-largest eigenvalues (in absolute value), respectively, of $HH^T$.
\label{parity_oriented_bnd}
\end{theorem}
\vspace{0.1in}
\begin{proof}

Let ${\bf q}=(q_1, \ldots,q_n)$ be a pseudocodeword of $G$, and let ${\bf p} = H{\bf q}$
be a  real-valued vector of length
$r$. The first eigenvector of $HH^T$ is ${\bf e}_1 = (1,1,\ldots,1)^T/ \sqrt{r}$. Let ${\bf p}_i$ be the projection
of ${\bf p}$ onto the $i$th eigenspace. We will now upper bound $\parallel H^T{\bf p} \parallel ^2$.
Converting $\parallel H^T{\bf p}\parallel ^2$ into eigenspace representation, we get
\[ \parallel H^T{\bf p}\parallel ^2 = \sum_{i = 1}^r \mu_i \parallel{\bf p}_i \parallel ^2 = \mu_1 \parallel {\bf
p}_1 \parallel ^2 + \sum_{i = 2}^r \mu_i \parallel{\bf p}_i \parallel ^2\]
\[ \le \mu_1 \parallel{\bf p}_1 \parallel ^2 + \mu_2 (\parallel{\bf p} \parallel ^2 -\parallel{\bf p}_1 \parallel ^2).
\]

Note that
\[ \parallel{\bf p}_1 \parallel ^2 = \frac{j^2}{r} (\sum_{i=1}^n q_i)^2, \mbox{ and }\]
\[ \parallel{\bf p} \parallel ^2 \le mj\Big{(}\sum_{i=1}^n q_i^2\Big{)}. \]
The first equality follows from the choice of ${\bf p}$ and the regularity of the parity check matrix $H$.
The second inequality follows by applying the identity $(q_1+q_2+\cdots+q_t)^2\le t(q_1^2+q_2^2+\cdots+q_t^2)$ to the
terms in the expansion of $\parallel {\bf p}\parallel^2$.

The above set of equations yield \[ \parallel H^T{\bf p} \parallel ^2 \le (\mu_1 - \mu_2)\frac{j^2}{r}\Big{(}\sum_{i = 1}^n
q_i\Big{)}^2 +
\mu_2 mj\Big{(}\sum_{i = 1}^n q_i^2\Big{)}.\]
We now lower bound $\parallel H^T{\bf p} \parallel ^2$ as follows
\[ \parallel H^T{\bf p} \parallel ^2 = \sum_{t = 1}^n \Big{(}\sum_{i = 1}^r \sum_{\ell = 1}^n h_{i,t} h_{i,\ell} q_{\ell}\Big{)}^2
\ge (4j^2)\Big{(}\sum_{t = 1}^n q_t^2\Big{)}.\]

This bound may be seen by observing that for each $t$ in the outer summation,
the inner sums over the indices $i$ and $\ell$ contribute $j$ $q_t$ terms
and
$(m-1)j$ terms involving other $q_k$'s.  When $t$ is fixed, for each $i$ wherein $h_{it}=1$,
we have $q_t$ and $(m-1)$ other $q_k$'s that contribute to the inner sum. Since $q_t$ and the $(m-1)$
other $q_k$'s are involved in the $i$th constraint node and since ${\bf q}$ is a pseudocodeword,
we have $q_t$+sum of $(m-1)$ other $q_k$'s $\ge 2q_t$. Since there are $j$ values of $i$ wherein
$h_{it}=1$, for a fixed $t$, the inner sum over $i$ and $\ell$ can be lower bounded by $2jq_t$.
Thus, $ \parallel H^T{\bf p} \parallel ^2 \ge \sum_{t = 1}^n
(2jq_t)^2=4j^2 (\sum_{t=1}^n q_t^2)$.

Combining the upper and lower bounds, we get
\[ \frac{(4j^2-\mu_2 mj)r}{(\mu_1 - \mu_2)j^2} \le \frac{(\sum_{i = 1}^n q_i)^2}{(\sum_{i = 1}^n q_i^2)} =
w^{AWGN}({\bf q}).\]
Since $nj = rm$, we obtain the desired lower bound.
\end{proof}

\begin{remark}
Note that this lower bound is not as strong as the bit-oriented bound in \cite{vo04}. It equals the bit-oriented bound for the case when
$m=2$. However, we believe that by a different but judicious choice of ${\bf p}$
in the above proof and by using stronger intermediate bounding steps, a
much stronger parity-oriented bound can be obtained.
\end{remark}

\section{Conclusions}
In this paper, the expander-based (i.e., eigenvalue-type) lower bounds on the minimum distance of expander codes were
extended to lower bound the minimum stopping set size and the minimum pseudocodeword weight of these codes. A new
parity-oriented lower bound in terms of the eigenvalues of the parity-check matrix was also obtained for the minimum
pseudocodeword weight of LDPC codes on the AWGN channel. These lower bounds indicate that LDPC codes constructed from
expander graphs provide a certain guaranteed level of performance and error-correction capability with graph-based iterative decoding
as well as linear programming decoding. Further, the results indicate that if the underlying LDPC constraint graph is a good expander, then the corresponding
expander code has a  minimum BSC pseudocodeword weight that is linearly growing in the block length. This is in general
a very hard criterion to ensure in the construction of good error correcting codes at large block lengths.
It would be interesting to derive upper bounds
on the distance, stopping set size, and pseudocodeword weight of expander codes to examine how tight the derived lower bounds are.

\section*{Acknowledgments} We thank Joachim Rosenthal and the reviewers for a careful proof-reading of this paper and their valuable comments. We believe their feedback has greatly improved
the paper. We also thank Reviewer 1 for providing the more intuitive definition of a stopping set in a generalized LDPC code.

\medskip
Received September 2006; revised July 2007.

\medskip
 {\it E-mail address: }ckelley@math.ohio-state.edu\\
 \indent{\it E-mail address: }deepak.sridhara@seagate.com\\

\end{document}